\documentclass[]{article}
\usepackage[affil-it]{authblk}

\usepackage{fullpage}
\usepackage{pdflscape}
\usepackage{secdot}

\usepackage{latexsym}
\usepackage{booktabs}
\usepackage[numbers]{natbib}
\usepackage[usenames,dvipsnames]{color}

\usepackage{epsfig}
\usepackage{subcaption}
\usepackage{calc}
\usepackage{amssymb}
\usepackage{amstext}
\usepackage{amsmath}
\usepackage{amsthm}
\usepackage{multicol}
\usepackage{pslatex}

\usepackage[acronym]{glossaries}
\usepackage[table]{xcolor}
\usepackage[ruled,vlined]{algorithm2e}
\usepackage{siunitx}

\renewcommand{\paragraph}[1]{{\vskip2mm\noindent \bf #1}\hspace{0.4cm}}

\definecolor{purple}{RGB}{153,50,204}

\newcommand{\ccb}{\cellcolor{black!30}}
\newcommand{\ccbb}{\cellcolor{black!40}}
\newcommand{\ccr}{\rowcolor{black!10}}

\newacronym{mcsp}{MCSP}{Minimum Common String Partition Problem}
\newacronym{dd}{DD}{Decision Diagram}
\newacronym{aco}{ACO}{Ant Colony Optimization}
\newacronym{tresea}{TRESEA}{Probabilistic Tree Search}
\newacronym{ilp}{ILP}{Integer Linear Programming}
\newacronym{dp}{DP}{Dynamic Program}

\newtheorem{theorem}{Theorem}
\newtheorem{example}{Example}

\begin{document}

\title{On Solving the Minimum Common String Partition Problem by Decision Diagrams}

\author[]{Milo\v{s} Chrom\'{y}\thanks{milos.charomy@jku.at}}
\author[]{Markus Sinnl\thanks{markus.sinnl@jku.at}}

\affil[]{Institute of Production and Logistics Management/JKU Business School, Johannes Kepler University Linz, Linz, Austria}

\date{}

\maketitle

\begin{abstract}
In the \acrfull{mcsp}, we are given two strings on input, and we want to partition both into the same collection of substrings, minimizing the number of the substrings in the partition. This combinatorial optimization problem has applications in computational biology and is NP-hard. Many different heuristic and exact methods exist for this problem, such as a Greedy approach, \acrlong{aco}, or \acrlong{ilp}. In this paper, we formulate the \acrshort{mcsp} as a \acrlong{dp} and develop an exact solution algorithm based on \acrlong{dd}s for it. We also introduce a restricted \acrlong{dd} that allows to compute heuristic solutions to the \acrshort{mcsp} and compare the quality of solution and runtime on instances from literature with existing approaches. Our approach scales well and is suitable for heuristic solution of large-scale instances.
\end{abstract}

\section{\uppercase{Introduction}}
A string is a sequence of symbols such as letters of the English alphabet, numbers, or nucleotides ($\mathbf{ACGT}$) forming a DNA sequence. Many optimization problems related to strings are widespread in bioinformatics, such as the Far-from Most String Problem~\cite{meneses_optimization_2005,mousavi_improved_2012}, the Longest Common Subsequence Problem and its variants~\cite{hsu_computing_1984,smith_identification_1981}, and Sequence Alignment Problems~\cite{gusfield_algorithms_1997}. In this paper we focus on the~\emph{\acrfull{mcsp}}. 
In the \acrshort{mcsp} we are given two or more related input strings, where related means they contain the same symbols. A solution is the partitioning of each input string into the same collection of substrings. The objective is to minimize the size of the obtained collection.

\begin{example}
\label{ex:mcsp}
Consider the \acrshort{mcsp} on two DNA sequences $s_1 = \mathbf{GAGACTA}$ and $s_2=\mathbf{AACTGAG}$. Obviously, $s_1$ and $s_2$ are related because $\mathbf{A}$ appears three times in both input strings, $\mathbf{G}$ appears twice in both input strings, and $\mathbf{C}$ and $\mathbf{T}$ appear only once. A trivial valid solution can be obtained by partitioning both strings into substrings of length one  $P_1 = P_2 = \{\mathbf{A, A, A, C, T, G, G}\}$. The objective function value of this solution is 7. However, the optimal solution with objective function value 3 is $P_1 = P_2 = \{\mathbf{A, ACT, GAG}\}$.
\end{example}

The \acrshort{mcsp} is closely related to the \emph{Problem of Sorting by Reversals with Duplicates}, a key problem in genome rearrangement~\cite{chen_assignment_2005}.

The work of \cite{goldstein_minimum_2005} proved the NP-hardness of the \acrshort{mcsp}. Different heuristics for the \acrshort{mcsp} were introduced, such as a \emph{Greedy approach}~\cite{chrobak_greedy_2004,he_novel_2007}, \emph{\acrfull{tresea}}~\cite{blum_iterative_2014}, \emph{\acrfull{aco}}~\cite{ferdous_solving_2017}. Moreover, also various exact approaches based on \emph{\acrfull{ilp}} models~\cite{blum_minimum_2020,blum_mathematical_2015,blum_construct_2016,blum_computational_2016} were proposed.

\paragraph{Our contribution.}
In this work, we develop a \emph{\acrfull{dd}} approach to the \acrshort{mcsp}.
A \acrshort{dd} can be used for computing the exact solution of a chosen optimization problem. However, NP-hard problems such as the \acrshort{mcsp} could have an exponentially large \acrshort{dd} representation. By \emph{relaxing} and \emph{restricting} \acrshort{dd}s, lower and upper bounds for the objective function value of a optimization problem can be obtained~\cite{bergman_decision_2016}. Such \acrshort{dd} approaches were  already used for various combinatorial problems such as Graph Coloring, Maximal Independent Set, MaxSAT~\cite{bergman_decision_2016,van_hoeve_graph_2020}, and various sequence and string problems such as Repetition-Free Longest Common Subsequence~\cite{horn_use_2020}, Multiple Sequence Alignment~\cite{hosseininasab_exact_2021}, or Constraint-Based Sequential Patter Mining~\cite{hosseininasab_constraint-based_2019}.
\acrshort{dd}s are also suitable for a hybrid approach with \acrshort{ilp} solvers combined using Machine Learning methods~\cite{gonzalez_integrated_2020,tjandraatmadja_decision_2018}.

This paper introduces an exact and a restricted \acrshort{dd} for the \acrshort{mcsp}. 
%
We compare their performance for solving the \acrshort{mcsp} with existing results achieved by~\acrshort{aco}, \acrshort{tresea}, and \acrshort{ilp} models on the dataset  used in~\cite{blum_mathematical_2015}, \cite{blum_iterative_2014,blum_construct_2016} and~\cite{ferdous_solving_2017}. 

Our goal is to motivate further research by comparing the \acrshort{dd} approach with existing heuristic approaches for the \acrshort{mcsp}. We show that our restricted \acrshort{dd} allows computation of better objective function values than \acrshort{aco}~\cite{ferdous_solving_2017} for most instances, and better objective function values than \acrshort{tresea}~\cite{blum_iterative_2014} for large instances. Moreover, we show that our restricted \acrshort{dd} approach is much faster than any of the currently applied methods (\acrshort{aco}, \acrshort{tresea}, \acrshort{ilp}~\cite{blum_mathematical_2015,blum_iterative_2014,blum_construct_2016,blum_computational_2016,ferdous_solving_2017}). 


\paragraph{Outline of the paper.}

The paper is organized as follows. In the Section~\ref{sec:def}, we go through basic definitions used in the paper. In Section~\ref{sec:ddmodel}, we introduce a \emph{\acrfull{dp}} formulation of the \acrshort{mcsp}. This formulation allows us to define exact and restricted \acrshort{dd}s for the \acrshort{mcsp}, which we also do in this section. We describe the main implementation details in the Section~\ref{sec:ddmodel}. In Section~\ref{sec:experiments}, we present our experimental results, and in the last Section~\ref{sec:concl}, we summarize the results and set up goals for future research.

\section{\uppercase{Definitions}}
\label{sec:def}
First, we define a notion of a Minimization Problem, the \acrlong{dp} and \acrlong{dd} for a Minimization Problem, following the definitions given in \cite{bergman_decision_2016,tjandraatmadja_decision_2018}. Finally, we also formally define the \acrlong{mcsp}. 

\subsection{Minimization Problem}
We consider the Boolean \emph{Minimization Problem} $P$ over $n$ Boolean variables, where we try to minimize the function value $f(x)$ according to the constraints $C_i(x)$, $i=1,\ldots,m$ and $x\in\{0,1\}^n$. Constraints $C_i(x)$ state an arbitrary relation between two or more variables. A \emph{feasible solution} to $P$ is an assignment $x\in\{0,1\}^n$ satisfying all constraint $C_i(x)$, $i=1,\ldots,m$. The set $Sol(P)$ is the set of all feasible solutions to $P$.
A feasible solution $x^*\in Sol(P)$ is \emph{optimal} for $P$ if $f(x^*)\leq f(x)$ for all $x\in Sol(P)$. The set $Opt(P)$ is the set of all optimal solutions of problem $P$.

\subsection{\acrlong{dp}}
A \emph{\acrfull{dp}} for a given problem $P$ with $n$ Boolean variables consists of a state space $S$ with $n+1$ stages, a partial transition function $t_j$, and a transition cost function $h_j$. 
\begin{itemize}
\item
The \emph{state space S} is partitioned into sets for each of the $n+1$ stages; i.e., $S$ is the union of the sets $S_0,\ldots,S_{n}$, where $S_0$ contains only the \emph{root state} $\hat{r}$, and $S_{n}$ contains only the \emph{terminal state} $\hat{t}$.
\item 
The partial \emph{transition functions $t_j$} defines how the decisions govern the transition between states. Note that this functions is not necessarily defined for each state and decision, however each state has defined a transition for at least one decision. 
\item
The \emph{transition cost functions }$h_j$ is defined for each transition defined by $t_j$ and gives the cost for taking this transition.
To account for objective function constants, we also consider a \emph{root value $v_r$} which is constant that will be added to the transition costs directed out of the root state.
\end{itemize}

The \acrshort{dp} formulation has variables $$(s,x)=(\hat{r}=s^0,\ldots,s^n=\hat{t},x_0,\ldots, x_{n-1}).$$ 
The objective function value we try to minimize is $$\hat{f}=v_r+\sum_{j=0}^{n-1}h_j(s^j,x_j)$$
with  $$s^{j+1}=t_j(s^j,x_j)\quad s^j\in S_j, j=0,\ldots,n-1$$ 
for a defined transition function $t_j$. 

This formulation is \emph{valid} for a problem $P$ if, for every $x\in Sol(P)$ there is an $s\in S_0\times S_1\times \ldots \times S_n$ such that $(s,x)$ is feasible and $\hat{f}(s,x)=f(x)$.

\subsection{\acrlong{dd}s}
A \emph{\acrfull{dd}} $D=(U,A,h)$ is a arc-weighted layered directed acyclic multigraph with node set $U$, arc set $A$, weight function $h:A\rightarrow \mathbb{R}$, and arc labels $0$ and $1$. The node set is partitioned into layers $L_0,\ldots,L_n$, where layer $L_0$ contains only root node $r$ and $L_n$ contains only terminal $\mathbf{T}$. The \emph{width} of the layer is the number of the nodes it contains.  We also consider a weight $v_r\in\mathbb{R}$ of the root node $r$ for problem specific constants. Each node on layer $L_j$ is associated with a Boolean variable $x_j$.
Each arc $a\in A$ is directed from a node $n_j\in L_j$ to a node $n_{j+1}$ in the next layer $L_{j+1}$ and has a label $0$ for a $0$-arc or $1$ for $1$-arc that represents assignment a value $0$ or $1$ to a variable $x_j$. The node $n_j[x_j=0]$ denotes the endnode of the $0$-arc with a start node $n_j$, and the node $n_j[x_j=1]$ denotes the endnode of $1$-arc with a start node $n_j$. Every arc-specified path $p=(a^0,a^1,\ldots,a^{n-1})$ from $r$ to $\mathbf{T}$ encodes an assignment to the variables $x_0,\ldots,x_{n-1}$, namely $x_j=a^j$, $j=0,\ldots,n-1$. The weight of such path is $h(p)=v_r + \sum_{j=0}^{n-1}h(a^j)$. 
The set of $r$-$\mathbf{T}$ paths of $D$ represents the set of assignments $Sol(D)$. The set of all minimum weighted $r$-$\mathbf{T}$ paths of $D$ represents the set of optimal assignments $Opt(D)$.

 A \acrshort{dd} $D$ is \emph{exact} for problem $P$ if $Sol(D) = Sol(P)$. The benefit of using exact \acrshort{dd}s for representing solutions is that equivalent nodes, i.e. nodes with the same set of completions, can be merged. A \acrshort{dd} is called \emph{reduced} if no two nodes in a layer are equivalent. A key property is that for a given fixed variable ordering, there exists a unique reduced \acrshort{dd}. Nonetheless, even reduced \acrshort{dd}s may be exponentially large to represent all solutions for a given problem.

A \acrshort{dd} $D$ is \emph{restricted} for problem $P$ if $Sol(D) \subseteq Sol(P)$ and $h(x_D^*)\geq f(x_P^*)$ for $x_D^*\in Opt(D)$ and $x_P^*\in Opt(P)$. 
Such a restricted \acrshort{dd} can give us a heuristic solution of a minimization problem $P$ in short computing time and with manageable memory requirements.

\subsection{\acrlong{mcsp}}
An \emph{alphabet} $\Sigma=\{a_1,a_2,\ldots,a_{|\Sigma|}\}$ is a finite set of \emph{symbols}.
A \emph{string} $s$ is a finite sequence of symbols from $\Sigma$. The \emph{length} $n$ of a string $s$ is the number of symbols contained in $s$. A \emph{substring}  $s[i:j], 0\leq i<j\leq n$ denotes a substring of $s$ consisting of symbols of $s$ starting at index $i$ and ending on index $j-1$. The term $1^k$($0^k$) denotes a string of $k$ repeating ones (zeroes).

Two strings $s_1,s_2$ are \emph{related} iff each symbol appears the same number of times in each of them. 
A valid solution to the \acrshort{mcsp} is a partitioning of $s_1$ and $s_2$ into multisets $P_1$ and $P_2$ of non-overlapping substrings, such that $P_1=P_2$. The value of the solution to the \acrshort{mcsp} is the size of partitioning $|P_1|=|P_2|$ and the goal is to find a solution with the minimal value.

\section{\uppercase{Modeling the minimum common string partition problem AS DECISION DIAGRAM}}
\label{sec:ddmodel}

First, we describe the \acrshort{dp} formulation of the \acrshort{mcsp}. We then use the \acrshort{dp} formulation to describe the exact \acrshort{dd} formulation of the \acrshort{mcsp}. Lastly we describe a restricted \acrshort{dd} for the \acrshort{mcsp}.

\subsection{\acrshort{mcsp} as a \acrlong{dp}}
\label{ssec:dpmodel}
Consider the \acrshort{mcsp} on two related strings $s_1,s_2$ of length $n$. A \emph{block} $b_i$ is a tuple $(k_i^1,k_i^2,t_i)$ such that $k_i^1,k_i^2\in\{ 0,\ldots, n-t_i\}$, $t_i\in\{1,\ldots, n\}$ and substrings $s_1[k_i^1:k_i^1+t_i]$ and $s_2[k_i^2:k_i^2+t_i]$ are equal, i.e. contain same symbols in same order. Two blocks $b_i$, $b_j$ \emph{overlap} if at least one of following holds
 \begin{enumerate}
     \item $k_i^1 - t_j < k_j^1 < k_i^1 + t_i$ or
     \item $k_i^2 - t_j < k_j^2 < k_i^2 + t_i$.
 \end{enumerate}
 This means, that both $b_i$ and $b_j$ are associated with at least one same position $\ell=\{1,\ldots,n\}$ in the input strings $s_1$ and $s_2$. Otherwise the blocks do \emph{not overlap}.
 
 \begin{example}
 Consider the \acrshort{mcsp} on two DNA sequences $s_1 = \mathbf{\widehat{GA\underline{G}}\widetilde{\underline{A}CT}A}$ and $s_2=\mathbf{A\widetilde{ACT}\widehat{\underline{G}\underline{A}G}}$ as in Example~\ref{ex:mcsp}. We consider blocks $b_i=(2,5,2)$ associated with the string $\mathbf{\underline{GA}}$, $b_j=(3,1,3)$ associated with the string $\mathbf{\widetilde{ACT}}$, and block $b_\ell=(0,4,3)$ associated with the string $\mathbf{\widehat{GAG}}$. String $b_i$ and $b_j$ overlap only at position $3$ in the first string $s_1$. Blocks $b_i$ and $b_\ell$ overlap in both strings in positions $s_1[2]$ and $s_2[5:7]$. The blocks $b_j$ and $b_\ell$ does not overlap.
 \end{example}

Let us define the set of all blocks of length at least two as $B$ with associated index set $I$, i.e., each index $(k^1,k^2,t_b) \in I$ corresponds to a block in $B$. 

The variables of our \acrshort{dp} formulation of the \acrshort{mcsp} are indexed by the index set $I_\lambda=I \cup \{\lambda\}$ where $I$ is the index set of our blocks and $\lambda\not\in I$ is a special index. The set of variables is defined as $X=\{x_i |i\in I_\lambda\}$. Each variable $x_i$, $i\in I$ is associated with the corresponding block and the variable $x_\lambda$ is a variable, which is used in the last step in the \acrshort{dp} to indicate the covering of the remaining uncovered symbols by blocks of size one. The constrains of the \acrshort{mcsp} are defined as
$\neg (x_i\wedge x_j)$ for each $i,j$ such that blocks $b_i$, $b_j$ overlap. 

A state of our \acrshort{dp} is defined by a tuple of bitsrings $(bs^1,bs^2)$, where $bs^1,bs^2\in\{0,1\}^n$. Let $B_\ell \subseteq B$ be the set of all blocks previously considered when we are in state $\ell$ of our \acrshort{dp}.
Let $bs^1[j]$ indicate the $j$-th position in $bs^1$. The value of $bs^1[j]$ is zero if and only if the position $j$ in $s_1$ is contained in a block $b_i\in B_\ell$ where we chose $x_i=1$. Similarly $bs^2[j]=0$ if and only if position $s_2[j]$ is contained in any block $b_i\in B_\ell$ where we chose $x_i=1$.

The root state is $(1^n,1^n)$ as no position is yet covered and the terminal state is $(0^n,0^n)$ as all positions have to be covered in a feasible solution.

The transition function for a state $(bs^1,bs^2)$ and a variable $x_i$, $i\in I_\lambda$ is

\begin{itemize}
    \item $(bs^1,bs^2)$ for $i\in I$, $x_i=0$, with the cost 0.
    \item $(bs_{new}^1,bs_{new}^2)$ where $c=0,1$ $bs_{new}^c[j_c]=0$ for all $k_i^c\leq j_c < k_i^c+t_i$ and $bs_{new}^c[j_c]=bs^c[j_c]$ otherwise. The cost of this transition is $1-t_i$. This transition is defined only if $i\in I$, $x_i=1$ and $bs^c[j_c]= 1$ for $k_i^c\leq j_c < k_i^c+t_i$ for both $c=1,2$, i.e 
    block $b_i$ does not overlap with any block $b_j\in B_i$ where we choose $x_j=1$.
    Otherwise the transition is undefined. 
    \item $(0^n,0^n)$ for $x_\lambda=1$. This transition cost is 0.
    \item For $x_\lambda=0$, the transition is undefined.
\end{itemize} 

The cost of root state $v_r=n$, where $n$ is the length of input string $|s_1|=|s_2|$. 
The cost of transition $t((bs_i^1,bs_i^2),x_i=1)$, $i\in I$ ``saves'' us $t_i$ blocks of length 1 and uses one block $b_i$ of length $t_i$ instead. 

\subsection{An exact \acrlong{dd} for the \acrshort{mcsp}}
\label{ssec:exactdd}

Now we can describe the formulation of the \acrshort{mcsp} as an exact \acrshort{dd} using the \acrshort{dp} formulation. 

For each state $(bs^1,bs^2)$ on the stage $S_i$, $i\in I_\lambda$, we will create a node $n_i$ in the \acrshort{dd} on layer $L_i$ associated with the variable $x_i$. The 0-arc and 1-arc connecting node $n_i$ with $n_i[x_i=0]$ and $n_i[x_i=1]$ are defined by the transition function of state $(bs^1,bs^2)$ and the weight of each arc is defined by the cost function of transition in the \acrshort{dp} model. The last layer contains only one terminal node $\mathbf{T}$ representing the terminal state $\hat{t}$ with no arcs leaving the terminal node. The weight of the root node $r$ is the same as the cost $v_r$ of the root state $\hat{r}$ of the \acrshort{dp} formulation.

\begin{theorem}
\emph{\acrshort{dd}} $D$ is an exact \emph{\acrshort{dd}} of the \emph{\acrshort{mcsp}} and $Opt(D)=Opt(\emph{\acrshort{mcsp}})$. 
\end{theorem}
\begin{proof}
First, we prove that for every solution of \acrshort{mcsp} exists exactly one $r$-$\mathbf{T}$ path in the \acrshort{dd} $D$ with the same weight as the objective value of the original solution. The optimality equivalence $Opt(D)=Opt(\text{\acrshort{mcsp}})$ immediately follows.

Let us have a set  of variables on an $r$-$\mathbf{T}$ path $p$ set to 1, $T=\{x_i|(n_i,n_{i+1})\in p$ is a $1$-arc, $i\in I \}$. Let us take any $x_j\in T$ with a block $b_j=(k_j^1,k_j^2,t_j)$. Now for any variable $x_i\in T$ such that $b_i\in B_j$ a block $b_i=(k_i^1,k_i^2,t_i)$ neither of overlapping conditions 
 \begin{enumerate}
     \item $k_i^1 - t_j < k_j^1 < k_i^1 + t_i$ or
     \item $k_i^2 - t_j < k_j^2 < k_i^2 + t_i$
 \end{enumerate}
holds.
Otherwise, the block $b_j$  overlaps with some other block already included in a partial partition, and hence such arc does not exist as the transition in the \acrshort{dp} defining the 1-arc is not defined. The $r$-$n_{|I|}$ path for any node $n_{|I|}$ on the last layer $L_{|I|}$ of $D$ represents a valid partial partition of both input strings and the $1$-arc from node $n_{|I|}$ works as a shortcut which includes all remaining uncovered symbols in the final partition using single symbol blocks. The weight of such path is $n + |T| -\sum_{x_i\in T} t_i$, which corresponds exactly to the number of blocks used in the cover defined by any $r$-$\mathbf{T}$ path as each block $b_i=(k_i^1,k_i^2,t_i)$ ``saves'' us $t_i$ blocks of the length $1$ and uses one of the length $t_i$ instead.

For other direction, let us suppose that a partition of input strings, with set of blocks bigger than one $B=\{b_j| j\in J\}$ for some $J\subseteq I$, does not correspond to any $r$-$n_{|I|}$ path in \acrshort{dd} $D$. Let us construct longest $r$-$n_{j_m}$ path $p$ for some $j_m\in J$ such that 0-arc from $n_i$ is in the path $p$, $i<j_m$ and $i\not\in J$, i.e. block $i$ is not used in the partition and 1-arc if $i<j_m$ and $i\in J$. Obviously $j_m\in J$ and $n_{j_m}$ has no outgoing 1-arc. However, the block $b_{j_m}$ is part of partition and any block $b_i$, $i\in J$ does not overlap with the block $b_{j_m}$ hence, the transition for a state connected to the node $n_{j_m}$ and value $x_{j_m}=1$ is defined which implies existence of outgoing 1-arc by which we get the contradiction.

As we have shown in the first part of our proof, the weight of any $r$-$\mathbf{T}$ path equals the size of the partition and therefore $Opt(D)=Opt(\text{\acrshort{mcsp}})$.
\end{proof}

\subsection{An restricted \acrlong{dd} for the \acrshort{mcsp}}
We get a restricted \acrshort{dd} $D^\prime$ by restricting the width of each layer of the original \acrshort{dd} $D$ by a given bound $W\in\mathbb{N}$. To achieve this bound we simply remove nodes with respect to a given criterion. In our implementation of $D'$, we delete nodes $n_i$ with the highest weight $r$-$n_i$ path $p$. This weight corresponds to a partition, which uses all variables $x_i$ such that for $n_i$ path $p$ chooses $1$-arc together with the number of ones in $bs_i^1$ or $bs_i^2$. By such restriction remove some $r$-$\mathbf{T}$ paths from $D$ and we do not add any new $r$-$\mathbf{T}$ paths. Moreover, we do not change the weight function $h$ of the $D$. Hence, $Sol(D^\prime)\subseteq Sol(D) = Sol(\acrshort{mcsp})$ and $x_{D^\prime}^* \geq x_D^* = x_{\acrshort{mcsp}}^*$ for $x_{D^\prime}^*\in Opt(D^\prime),x_D^*\in Opt (D)$ and $ x_{\acrshort{mcsp}}^*\in Opt(\acrshort{mcsp})$ follows.

\paragraph{Construction of $D'$.}
Algorithm~\ref{alg:restdd} describes the construction of $D'$ for two strings $s_1,s_2$ of length $n$.

\begin{algorithm}[!htb]
\DontPrintSemicolon
\LinesNumbered
\KwResult{The upper bound of the optimal objective function value of the \acrshort{mcsp}.}
\KwData{Input strings $s_1,s_2$. \\ \phantom{Data: } Layer size restriction $W$.}
 $X\leftarrow$ Compute blocks of ($s_1$, $s_2$)\;
 Sort($X$)\;
 $r\leftarrow$ node with $bs_r^1=bs_r^2=1^n$\;
 $L_0\leftarrow \{r\}$\;
 $c(r)\leftarrow v_r$\;
 \For{$i=0,\ldots,|X|-1$}{
  \For{$n_i\in L_i$}{
    \If{$n_i[x_i=0]\not\in L_{i+1}$}{
        Add $n_{i+1}=n_i[x_i=0]$ to $L_{i+1}$\;
    }
    Update weight $c(n_i[x_i=0])$ by $c(n_i)$\;
    \If{1-arc $(n_i,n_i[x_i=1])$ is defined}{
        \If{$n_i[x_i=1]\not\in L_{i+1}$}{
            Add $n_{i+1}=n_i[x_i=1]$ to $L_{i+1}$\;
        }
        Update weight $c(n_i[x_i=1])$ by $c(n_i)-t_i+1$\;
    }
  }
  Restrict $L_{i+1}$ by $W$
 }
 \For{$n_i\in L_{|X|}$}{
    Update weight $c(n_i[x_i=1]=\mathbf{T})$ by $c(n_i)$\;
  }
  \KwRet{$c(\mathbf{T})$}
 \caption{Computing the \acrshort{mcsp} upper bound of the optimal objective function value by constructing the restricted \acrshort{dd}.}
 \label{alg:restdd}
\end{algorithm}

 First we collect all blocks $b_i=(k_i^1,k_i^2,t_i)$ of length $t_i>1$. The number of blocks is $O(n^3)$. 

Next, we sort blocks. The ordering is given by the size $t_i$ of a corresponding block $b_i$ in descending order, i.e. $x_i<x_j$ iff $t_i>t_j$. The variable $x_0$ associated with the root node corresponds to a maximal block.

Then we construct $D'$ layer by layer in a BFS manner. The transition function $t_i$  of the \acrshort{dp} formulation described in Section~\ref{ssec:dpmodel} defines arcs and states of new nodes $n_i[x_i=c]$, $c=0,1$.

If the nodes $n_i[x_i=0]$ or $n_i[x_i=1]$ are defined and do not appear in layer $L_{i+1}$ yet, we add them (lines 8-9, 12-13). By storing the layer $L_{i+1}$ as a hash table, we get the complexity of search and modification in the time of hashing of the node state, which is linear with the size of the representation of a node $n_i$ $O(n=|bs_i^1|=|bs_i^2|)$. We cache only the last full layer and newly constructed layer, to save memory consumption. Therefore, for each node $n_i$ we save the minimal weight $r$-$n_i$ path. The weight changes only in steps on lines 10 and 14 only if it improves the current weight. This update is a constant time operation.

We want to keep only the best $W$ nodes on any layer, i.e. the $W$ nodes with the minimal weight $r$-$n_i$ path, which we have saved for each node. First we find the $W$ minimal-weight nodes on the newly constructed layer in a linear time $O(W)$ as the newly constructed layer size is at most $2W$. Then we select the $W$ minimal-weight nodes nodes to keep also in a linear time $O(W)$. By skipping this step, we get an exact \acrshort{dd} construction.

In the last step, we connect 1-arc of all nodes layer $L_{|X|}$ to the $\mathbf{T}$ terminal and set the weight of $\mathbf{T}$ terminal to the minimal weight $r$-$\mathbf{T}$ path. This can be done in the time $O(W)$ as the size of the last layer is $O(W)$, and the update is a constant time operation. 

The time complexity of our algorithm 
is $O(\text{Blocks generation}+\text{Blocks sort}+|X|*\text{Layer generation})=O(n^3+n^3\log n^3+Wn^4)=O(Wn^4)$.

We have also tried different variable ordering methods and different node ordering used to impose the layer width constraint. However, the objective function values given by other orderings were higher than the orderings described in this section.

\section{\uppercase{Experimental evaluation}}
\label{sec:experiments}

Our solution approaches are implemented in \verb!C++! and 
run on an Intel$^{\text{\tiny{\textregistered}}}$ Core$^{\text{\tiny{\texttrademark}}}$ i5-8250U at 1.6GHz with at most 100MB of memory used. Our implementation is single threaded.

\subsection{Datasets}
We are using the datasets introduced in~\cite{ferdous_solving_2017} and in~\cite{blum_construct_2016}. 

The first dataset was introduced in the experimental evaluation of \acrshort{aco}~\cite{ferdous_solving_2017}. This dataset consists of 30 artificial instances and 15 real-life instances of DNA strings with the alphabet size $|\Sigma|=4$. This benchmark is divided into four subgroups. The Group 1 consists of 10 artificial instances of length $n\leq 200$, the Group 2 consists of 10 artificial instances of length $200\leq n\leq 400$, the Group 3 consists of 10 artificial instances of length $400\leq n \leq 600$. The last group called Real consists of 15 real-life instances of length $200\leq n \leq 600$.

\begin{table}[htb]
    \centering
    \caption{Results for the instances introduced by ~\cite{ferdous_solving_2017}. Columns DD$_W$ contain the objective function values obtained by using Decision diagrams with different restrictions on layer width. The remaining columns contain the objective function values reported in~\cite{blum_construct_2016} for the solution approaches considered in this paper. A cell is grey if an approach considered in~\cite{blum_construct_2016}  obtains a better objective function value than all restricted \acrshort{dd}$_W$.}
    \resizebox{0.7\textwidth}{!}{
    \begin{tabular}{c|S S S | S S S S S S}
\text{id}&\text{DD}$_{10}$&\text{DD}$_{100}$&\text{DD}$_{1000}$&\text{greedy}&\text{ACO}&\text{TRESEA}&\text{ILP}$_\text{comp}$&\text{heurILP}&\text{CMSA} \\
\multicolumn{10}{l}{Real}\\
\ccr 1&92&88&88&93&87\ccbb &86\ccbb &78\ccbb &85\ccbb &78.9\ccbb \\
2&160&160&157&160&155\ccb &153\ccb &139\ccb &150\ccb &140\ccb \\
\ccr 3&120&117&116&119&116&113\ccbb &104\ccbb &112\ccbb &104.7\ccbb \\
4&164&162&161&171&164&156\ccb &144\ccb &158\ccb &143.7\ccb \\
\ccr 5&169&169&168&172&171&166\ccbb &150\ccbb &161\ccbb &152.9\ccbb \\
6&149&146&144&153&145&143\ccb &128\ccb &139\ccb &127.6\ccb \\
\ccr 7&138&136&134&135&140&131\ccbb &121\ccbb &132\ccbb &122.7\ccbb \\
8&136&131&129&133&130&128\ccb &116\ccb &123\ccb &118.4\ccb \\
\ccr 9&144&141&143&149&146&142&131\ccbb &139\ccbb &130.7\ccbb \\
10&148&147&144&151&148&143\ccb &131\ccb &144&131.7\ccb \\
\ccr 11&125&125&126&124\ccbb &124\ccbb &120\ccbb &110\ccbb &122\ccbb &111.9\ccbb \\
12&141&141&142&143&137\ccb &138\ccb &126\ccb &136\ccb &127.5\ccb \\
\ccr 13&177&175&172&180&180&172&156\ccbb &171\ccbb &158.6\ccbb \\
14&151&151&152&150&147\ccb &146\ccb &134\ccb &147\ccb &134\ccb \\
\ccr 15&155&153&153&157&160&152\ccbb &139\ccbb &148\ccbb &141.7\ccbb \\
\multicolumn{10}{l}{Group 1}\\
\ccr 1&44&44&41&46&42&42&41&42&41\\
2&53&51&50&54&51&48\ccb &47\ccb &48\ccb &47\ccb \\
\ccr 3&60&57&55&60&55&55&52\ccbb &54\ccbb &52\ccbb \\
4&45&45&45&46&43&43\ccb &41\ccb &43\ccb &41\ccb \\
\ccr 5&45&44&43&44&43&41\ccbb &40\ccbb &43&40\ccbb \\
6&45&46&45&48&42&41\ccb &40\ccb &41\ccb &40\ccb \\
\ccr 7&62&62&60&64&60&59\ccbb &55\ccbb &59\ccbb &56\ccbb \\
8&47&44&44&47&47&45&43\ccb &44&43\ccb \\
\ccr 9&47&46&45&42&45&43\ccbb &42\ccbb &48&42\ccbb \\
10&61&58&58&63&59&58&54\ccb &58&54\ccb \\
%
%
\multicolumn{10}{l}{Group 2}\\
\ccr 1&116&112&109&118&113&111&98\ccbb &108\ccbb &101.2\ccbb \\
2&114&116&115&121&118&114\ccb &106\ccb &111\ccb &104.6\ccb \\
\ccr 3&115&107&110&114&111&107&97\ccbb &105\ccbb &97.1\ccbb \\
4&118&115&112&116&115&110\ccb &102\ccb &111\ccb &102.5\ccb \\
\ccr 5&138&134&129&132&132&127\ccbb &116\ccbb &125\ccbb &117.8\ccbb \\
6&106&107&105&107&105&102\ccb &93\ccb &101\ccb &95.4\ccb \\
\ccr 7&101&99&99&106&98\ccbb &95\ccbb &88\ccbb &96\ccbb &89\ccbb \\
8&116&117&113&122&118&114&104\ccb &116&105.2\ccb \\
\ccr 9&123&117&116&123&119&113\ccbb &104\ccbb &112\ccbb &104.9\ccbb \\
10&103&100&99&102&101&97\ccb &89\ccb &94\ccb &89.8\ccb \\
\multicolumn{10}{l}{Group 3}\\
\ccr 1&180&175&175&181&177&171\ccbb &155\ccbb &173\ccbb &157.9\ccbb \\
2&178&174&175&173\ccb &175&168\ccb &155\ccb &165\ccb &157.5\ccb \\
\ccr 3&194&188&188&195&187&185\ccbb &166\ccbb &180\ccbb &167.3\ccbb \\
4&188&184&182&191&184&179\ccb &159\ccb &171\ccb &161.8\ccb \\
\ccr 5&170&172&168&174&171&162\ccbb &150\ccbb &164\ccbb &151.1\ccbb \\
6&162&164&168&169&160&162\ccb &147\ccb &155\ccb &149.3\ccb \\
\ccr 7&169&167&167&171&167&159\ccbb &149\ccbb &160\ccbb &147.8\ccbb \\
8&175&178&171&185&175&170\ccb &151\ccb &166\ccb &154.2\ccb \\
\ccr 9&175&174&172&174&172&169\ccbb &158\ccbb &169\ccbb &155.3\ccbb \\
10&163&163&161&171&167&160\ccb &148\ccb &160\ccb &149\ccb  \\ 
\multicolumn{10}{l}{}\\ \multicolumn{10}{l}{}\\\multicolumn{10}{l}{}\\\multicolumn{10}{l}{}\\\multicolumn{10}{l}{}
    \end{tabular}}
    
    \label{tab:ACOresults}
\end{table}

 \clearpage

\begin{table}[htb]
    \centering
     \caption{
     Results for the instances introduced by ~\cite{blum_construct_2016}, averaged by instances of the same configuration. Columns DD$_W$ contain the objective function values obtained by using Decision diagrams with different restrictions on layer width. The remaining columns contain the objective function values reported in~\cite{blum_construct_2016} for the solution approaches considered in this paper. A cell is grey if an approach considered in~\cite{blum_construct_2016}  obtains a  better objective function value than all restricted \acrshort{dd}$_W$.
     }
    \resizebox{0.65\textwidth}{!}{
   \begin{tabular}{c|S S S | S S S S S}
\multicolumn{8}{l}{linear}\\
\text{n}&\text{DD}$_{10}$&\text{DD}$_{100}$&\text{DD}$_{1000}$&\text{greedy}&\text{TRESEA}&\text{ILP}$_\text{comp}$&\text{heurILP}&\text{CMSA} \\
\multicolumn{8}{l}{$\Sigma= 4$}\\
\ccr 200&73.6&72&70.5&75&68.7\ccbb &63.5\ccbb &69\ccbb &63.7\ccbb \\
400&132.2&129.6&128.3&133.4&126.1\ccb &115.7\ccb &124.3\ccb &116.4\ccb \\
\ccr 600&182.6&180.4&179&183.7&177.5\ccbb &162.2\ccbb &174.1\ccbb &162.9\ccbb \\
800&238.2&235.5&233.1&241.1&232.7\ccb &246.8&229.1\ccb &212.4\ccb \\
\ccr 1000&285.7&284.5&284&287&280.4\ccbb &n/a&277.2\ccbb &256.9\ccbb \\
1200&333.9&330.9&329&333.8&330.4&n/a&324.8\ccb &303.3\ccb \\
\ccr 1400&383.8&382.3&378.9&385.5&378.9&n/a&373.1\ccbb &351\ccbb \\
1600&429.4&424.4&422.8&432.3&427.1&n/a&416.7\ccb &400.6\ccb \\
\ccr 1800&476.4&473.6&471.1&477.4&474.2&n/a&464.4\ccbb &445.4\ccbb \\
2000&521.2&519.4&512.9&521.6&520.7&n/a&512.7\ccb &494\ccb \\
\multicolumn{8}{l}{$\Sigma= 12$}\\
\ccr 200&124.5&123.3&122.2&127.3&122.1\ccbb &119.2\ccbb &123&119.2\ccbb \\
400&227.5&225.8&224.1&228.9&223.5\ccb &208.9\ccb &215.7\ccb &209.4\ccb \\
\ccr 600&320.8&318.1&317.9&322.2&318.7&291\ccbb &296.2\ccbb &293.8\ccbb \\
800&412.8&409.3&405.3&411.4&408.1&368.7\ccb &373.9\ccb &373.2\ccb \\
\ccr 1000&497&494.1&491.5&499.2&494.9&453.4\ccbb &452\ccbb &449.9\ccbb \\
1200&586.2&581.2&578.5&586&585.6&536.6\ccb &542.4\ccb &531\ccb \\
\ccr 1400&663.7&658.1&657.9&666&664.6&684.1&653.3\ccbb &606.9\ccbb \\
1600&751.7&749.3&744.1&754.4&754.6&773.5&749.7&694.8\ccb \\
\ccr 1800&829.5&826.3&820.3&827.3&833&n/a&850.7&773.6\ccbb \\
2000&913.1&909.1&904.6&913.5&916.2&n/a&939.6&849.6\ccb \\
\multicolumn{8}{l}{$\Sigma= 20$}\\
\ccr 200&147.7&146.7&146.1&149.2&146.6&146.2&146.4&146.2\\
400&272.4&270.4&269.6&274.5&268.8\ccb &261.5\ccb &263.8\ccb &261.9\ccb \\
\ccr 600&387.9&385.4&383&389.2&383.5&362.3\ccbb &369.3\ccbb &366.6\ccbb \\
800&493&489.9&487.5&495.8&492.3&456.1\ccb &464.7\ccb &463.1\ccb \\
\ccr 1000&600&595.7&594&600.6&597.5&547.1\ccbb &562.5\ccbb &555\ccbb \\
1200&705.5&701.2&697.4&706.1&707.8&642.2\ccb &658.8\ccb &648.5\ccb \\
\ccr 1400&800.5&796.2&793.4&801.1&804&737.9\ccbb &745.7\ccbb &737.7\ccbb \\
1600&902.6&898.3&896.5&899.8&903.1&861.3\ccb &872.6\ccb &825.7\ccb \\
\ccr 1800&997.2&992.6&988.7&996.8&1000.1&1012.9&994.4&917.6\ccbb \\
2000&1095.8&1094.2&1086.4&1097.8&1102.6&1136&1120.7&1024.9\ccb \\
\multicolumn{8}{l}{skewed}\\
\text{n}&\text{DD}$_{10}$&\text{DD}$_{100}$&\text{DD}$_{1000}$&\text{greedy}&\text{TRESEA}&\text{ILP}$_\text{comp}$&\text{heurILP}&\text{CMSA} \\
\multicolumn{8}{l}{$\Sigma= 4$}\\
\ccr 200&67.8&66.2&64.2&68.7&62.8\ccbb &57.4\ccbb &64.6&57.5\ccbb \\
400&120.5&118.9&117.2&120.3&115\ccb &105.3\ccb &116.5\ccb &105.1\ccb \\
\ccr 600&169.5&167.6&165.4&170.6&163.8\ccbb &149.7\ccbb &165.2\ccbb &150.4\ccbb \\
800&220&216.2&214.7&219.8&213.3\ccb &224&211.7\ccb &196.5\ccb \\
\ccr 1000&265.8&265.8&263.2&268.6&261.7\ccbb &n/a&260.1\ccbb &240.2\ccbb \\
1200&315.9&311.5&309.3&313.8&309\ccb &n/a&302.1\ccb &285\ccb \\
\ccr 1400&362.3&357.4&354.3&358.7&352.2\ccbb &n/a&346\ccbb &327.6\ccbb \\
1600&403.7&401.4&397.1&400.9&397.9&n/a&394.4\ccb &376\ccb \\
\ccr 1800&444.1&438.4&438.5&440.6&442.1&n/a&431.7\ccbb &417.7\ccbb \\
2000&486.1&482.6&476.9&485&481.2&n/a&468.9\ccb &470.2\ccb \\
\multicolumn{8}{l}{$\Sigma= 12$}\\
\ccr 200&115.5&114.3&114&117.9&112.7\ccbb &108.5\ccbb &112.7\ccbb &108.6\ccbb \\
400&214.1&210.8&210.3&216.1&208.5\ccb &193.4\ccb &197.6\ccb &194.3\ccb \\
\ccr 600&305.1&304.3&300.4&304.8&301.7&274.5\ccbb &277.9\ccbb &277.2\ccbb \\
800&388.8&385.4&381.9&389.3&385.4&347\ccb &348.8\ccb &351\ccb \\
\ccr 1000&469.4&468.6&464&471.6&468.9&429.4\ccbb &428.7\ccbb &424.4\ccbb \\
1200&550.5&545.7&544.1&551.1&549.9&559.4&535\ccb &500.1\ccb \\
\ccr 1400&626.7&621.7&622.7&625.7&626.3&645.1&638.4&570\ccbb \\
1600&704.1&699.7&694.5&705.6&706.4&n/a&715.1&643.8\ccb \\
\ccr 1800&786.4&784.5&781.7&788.4&788.9&n/a&810.1&723.3\ccbb \\
2000&859.2&854.3&848.8&857.8&858&n/a&879.9&797.3\ccb \\
\multicolumn{8}{l}{$\Sigma= 20$}\\
\ccr 200&138.1&136.8&135.6&140.4&135.9&134.7\ccbb &136.5&134.7\ccbb \\
400&255&253&251.7&255.5&251.3\ccb &240.3\ccb &246.1\ccb &240.6\ccb \\
\ccr 600&364.3&362.6&360.3&366.8&361.2&336.1\ccbb &344.6\ccbb &341.1\ccbb \\
800&465.4&461.1&459.5&466.3&462.7&424.4\ccb &429.9\ccb &429.8\ccb \\
\ccr 1000&569.2&565.4&561.8&567.6&566.6&514.7\ccbb &525\ccbb &520.9\ccbb \\
1200&661.1&657.9&656&661.8&662.4&604.2\ccb &608.2\ccb &605.7\ccb \\
\ccr 1400&756.9&755&751.7&762.3&760.7&694.4\ccbb &696.1\ccbb &693.2\ccbb \\
1600&850&847.1&844.6&851.2&855.2&863.3&838.9\ccb &780.4\ccb \\
\ccr 1800&946.4&941.2&937.5&948.7&948.8&969.8&964.7&870.2\ccbb \\
2000&1033.9&1029.1&1025.6&1034.3&1037.7&1061.6&1066.6&967.1\ccb 
    \end{tabular}
    }
   
    \label{tab:CSMAdataset}
\end{table}

 \clearpage

\begin{figure}[htb]
    \centering
    \includegraphics[width=0.42\textwidth]{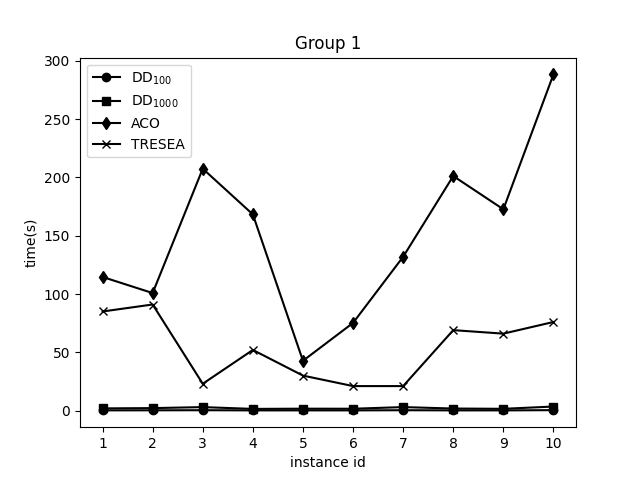}
    \includegraphics[width=0.42\textwidth]{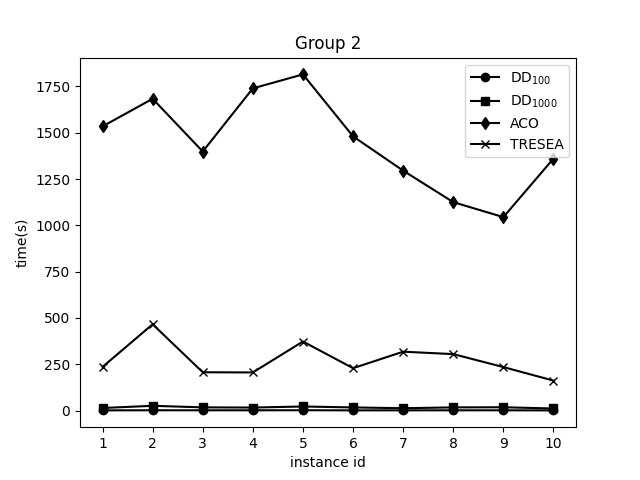}
    \includegraphics[width=0.42\textwidth]{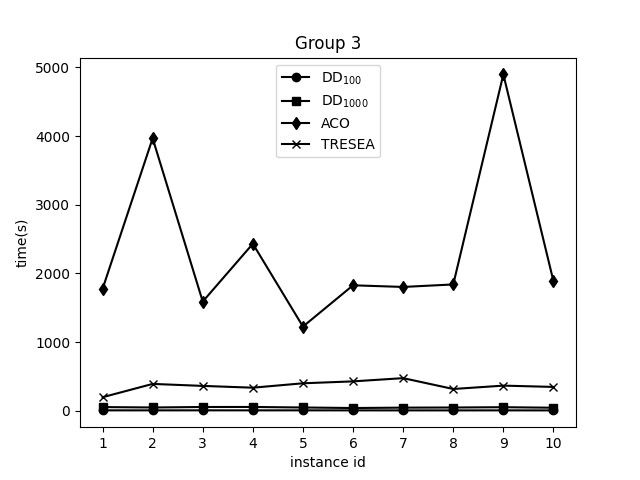}
    \includegraphics[width=0.42\textwidth]{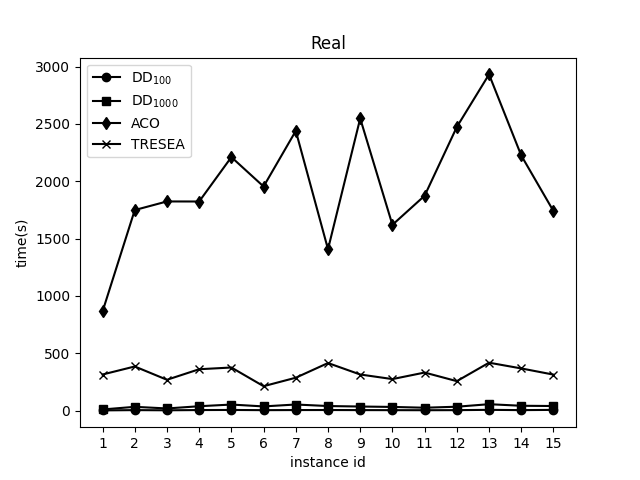}
    \caption{Runtime of the \acrshort{dd}$_W$ experiments with comparison to \acrshort{aco} and \acrshort{tresea}.}
    \label{fig:time_comparision_old}
\end{figure}

\begin{figure}[htb]
    \centering
    \includegraphics[width=0.42\textwidth]{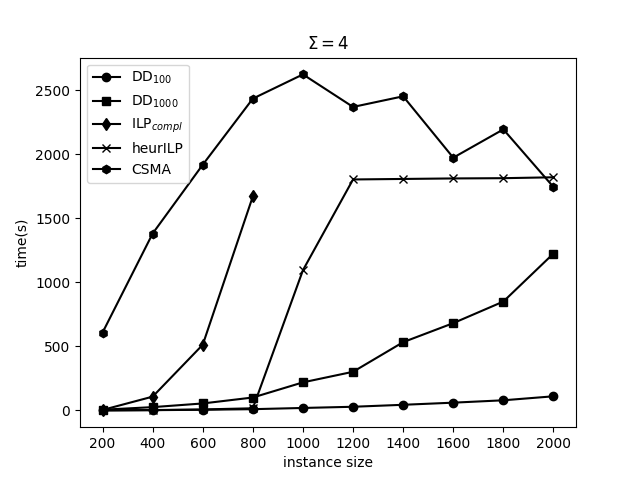}
    \includegraphics[width=0.42\textwidth]{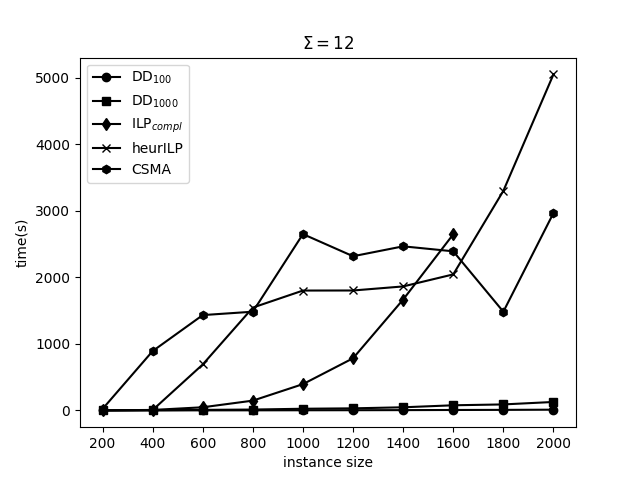}
    \includegraphics[width=0.42\textwidth]{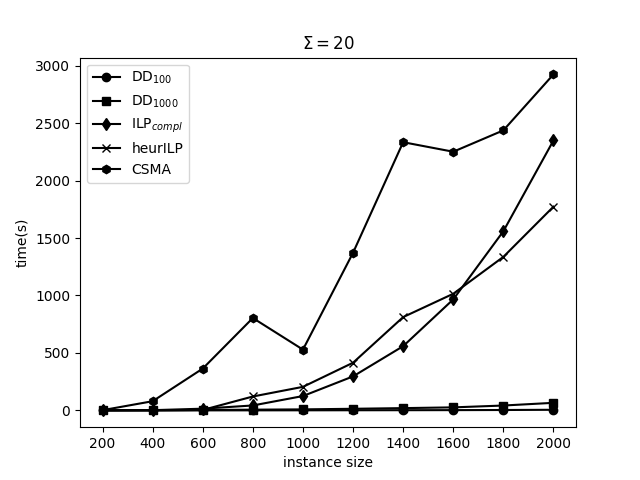}
    \caption{Runtime of \acrshort{dd}$_W$ experiments with comparison to the \acrshort{ilp} models on instance-set linea with $|\Sigma|=4,12,20$. The line for ILP$_{compl}$ ends sooner as it did not return any feasible solution after an hour for missing lines.}
    \label{fig:time_comparision}
\end{figure}


The dataset by~\cite{blum_construct_2016} consists of 20 randomly generated instances for each combination of length $n=200,400,\ldots,2000$ and alphabet size $|\Sigma|=4,12,20$. Ten of these instances are generated with the same probability for each symbol of the alphabet. These instances are called \emph{linear}. The remaining ten instances are called \emph{skewed}, and the probability of each symbol $l\in \Sigma$ with index $i\in\{1,\ldots,|\Sigma|\}$ is $i/\sum_{i=1}^{|\Sigma|} i$.

\subsection{Results}

We compare our results with the Greedy approach~\cite{chrobak_greedy_2004,he_novel_2007}, \acrfull{aco}~\cite{ferdous_solving_2017}, \acrfull{tresea}~\cite{blum_iterative_2014}, and various \acrshort{ilp} models (ILP$_{compl}$, heurILP, CMSA)~\cite{blum_mathematical_2015,blum_construct_2016}.
For these existing approaches we use results from~\cite{blum_construct_2016}, which were obtained using an Intel$^{\text{\tiny{\textregistered}}}$ Xeon$^{\text{\tiny{\texttrademark}}}$ X5660 CPU with 2 cores at 2.8 GHz and 48 GB of RAM.

 We have performed experiments with three restrictions of layer width of \acrshort{dd}s.
 \begin{itemize}
     \item \acrshort{dd}$_{10}$ has layer width at most 10, 
     \item \acrshort{dd}$_{100}$ has layer width at most 100, and
     \item \acrshort{dd}$_{1000}$ has  layer width at most 1000.
 \end{itemize}
 The \acrshort{dd}$_1$ with layer width 1 is the Greedy approach.
 We also did experiments with the exact \acrshort{dd}. The runtime of the exact \acrshort{dd} approach exceeds the time limit even for smallest instances and hence we do not report the results in our article.

The results of experiments on the dataset~\cite{ferdous_solving_2017} are presented in Table~\ref{tab:ACOresults}. Each line contains the objective function values for an instance obtained by different solution approaches.
The results for the second dataset~\cite{blum_construct_2016} are presented in Table~\ref{tab:CSMAdataset}. Each line contains the average objective function value obtained for all instances of the same configurations.

As we can see in Figure~\ref{fig:time_comparision_old}, the \acrshort{dd} approach runs faster than \acrshort{aco} and in most cases computes a better \acrshort{mcsp} objective function values, as can be seen in Table~\ref{tab:ACOresults}. 
The \acrshort{aco} runs at least $30\times$ slower on 93\% of instances. The runtime difference is over 2000 seconds for 20\% of instances.

As we can see in Tables~\ref{tab:ACOresults} and~\ref{tab:CSMAdataset}, \acrshort{tresea}  and the \acrshort{dd} approach behave similarly. For smaller instances of the \acrshort{mcsp}, \acrshort{tresea} outperforms the restricted \acrshort{dd}. However, for instances of the \acrshort{mcsp} on longer strings on alphabets $|\Sigma|\in\{12,20\}$ the \acrshort{dd} approach computes better objective function value on the \acrshort{mcsp} than the \acrshort{tresea} approach. As we can see in Figure~\ref{fig:time_comparision_old} the \acrshort{dd} even has a better runtime. \acrshort{tresea} is at least $6\times$ slower for 93\% of instances than the \acrshort{dd} approach.

The Greedy approach~\cite{chrobak_greedy_2004,he_novel_2007} is very similar to our approach of \acrshort{dd}$_1$. As we can see in both Tables~\ref{tab:ACOresults} and~\ref{tab:CSMAdataset} the restricted \acrshort{dd}s with bigger layer width gives better bounds.

The main advantage of the \acrshort{dd} is the runtime compared with basic the \acrshort{ilp} models. The runtime needed to get objective function values from the \acrshort{ilp} models grows fast with the instance size in comparison with the~\acrshort{dd}. This can be seen in both tables~\ref{tab:ACOresults} and~\ref{tab:CSMAdataset}. The \acrshort{dd} get worse bounds than \acrshort{ilp} in case of enough time. However, as the instances grow, the \acrshort{dd} gets better objective function values in a given time.
The \acrshort{dd} approach is $50\times$ faster for 50\% of instances then the CSMA approach.

\section{\uppercase{Conclusion}}
\label{sec:concl}

In this work, we developed a \acrshort{dp} formulation for the \acrshort{mcsp}. Based on this \acrshort{dp} formulation we designed an exact \acrshort{dd} solution approach and a heuristic solution approach using a restricted version of this \acrshort{dd}. The exact \acrshort{dd} is not suitable for solving large instances of the \acrshort{mcsp} as the runtime grows exponentially. The restricted \acrshort{dd} scales much better and can be used to heuristically solve the \acrshort{mcsp} with much larger input. We use the datasets from literature~\cite{blum_construct_2016,ferdous_solving_2017} to compare our \acrshort{dd} approach with  existing approaches such as Greedy approach~\cite{chrobak_greedy_2004,he_novel_2007}, \acrshort{aco}~\cite{ferdous_solving_2017}, \acrshort{tresea}~\cite{blum_iterative_2014} and various \acrshort{ilp} based models~\cite{blum_mathematical_2015,blum_construct_2016}.

In our experiments we have shown our approach is better than the Greedy approach and \acrshort{aco} both in the obtained objective function values and the runtime. Moreover, it obtains better objective function values for the \acrshort{mcsp} with bigger alphabets and longer strings than \acrshort{tresea}, ILP$_{compl}$ and heruILP in a given time. The CSMA obtains better results than our approach, however, it is much slower. 

A potential next research step could be improving the heuristic for variable ordering and node ordering of the restricted \acrshort{dd} together with a better heuristic for the layer width restriction. Another avenue for further research could be to hybridize the \acrshort{dd} model with the \acrshort{ilp} model. The \acrshort{dd} model runs fast and its simplicity allows to include any heuristic in any node to serve as the criterion to delete such node with its descendants and get the restricted \acrshort{dd} computing better objective function values. We hope that it could compete with the CSMA approach.

\subsection*{Future Work}

Here we present some additional ideas in which way future research could be directed.

\begin{enumerate}
    \item Current implementation builds the \acrshort{dd} breadth-first. We could change the order of processed nodes by using the priority queue instead processing whole layer at once. The node order can be for example 
    \begin{itemize}
        \item combination of the value of $r$-$n_i$ path and the number of yet uncovered symbols (A$^*$ approach),
        \item bounds by more restricted or relaxed \acrshort{dd} or by \acrshort{ilp} on subproblem defined by node $n_i$. The decision to use either \acrshort{dd} or \acrshort{ilp} for bounds can be made by Machine learning methods as Gonzales et al. used for a Maximum independent set problem~\cite{gonzalez_integrated_2020}.
    \end{itemize}
    \item Change the strategy of choosing node deleted during the restriction step. Our implementation uses a linear approach with the newly created layer's size, significantly speeding up the algorithm. We can change the speed for some more complex valuation of nodes.
    \item In our work, we only implemented a restricted \acrshort{dd} for the \acrshort{mcsp}, which gives us the upper bound of the optimal objective function value. For the lower bound, we can use a relaxed \acrshort{dd}~\cite{bergman_decision_2016}. The main advantage of such an approach would be that the relaxed \acrshort{dd} can be incrementally refined to get better solutions.
    We have already tried a few relaxations. However, we obtained a very weak lower bound on the \acrshort{mcsp}. We are interested in a suitable method for relaxation yielding better lower bounds in a reasonable time.
\end{enumerate}

\section*{\uppercase{Acknowledgements}}

This work was supported by the JKU Business School.

\bibliographystyle{plain}
\bibliography{references}

\end{document}